\documentclass{article}

\def\noheaderplainsetup{

\topmargin=0pt \headheight=0pt \headsep=0pt  \oddsidemargin=0pt \evensidemargin=0pt  \textheight=8.9truein \textwidth=6.2truein}

\noheaderplainsetup

\usepackage{amsfonts}
\usepackage{cite}
\begin{document}


\newcommand{\code}[1]{\ulcorner #1 \urcorner}
\newcommand{\mldi}{\hspace{2pt}\mbox{\footnotesize $\vee$}\hspace{2pt}}
\newcommand{\mlci}{\hspace{2pt}\mbox{\footnotesize $\wedge$}\hspace{2pt}}
\newcommand{\emptyrun}{\langle\rangle} 
\newcommand{\oo}{\bot}            
\newcommand{\pp}{\top}            
\newcommand{\xx}{\wp}               
\newcommand{\legal}[2]{\mbox{\bf Lr}^{#1}_{#2}} 
\newcommand{\win}[2]{\mbox{\bf Wn}^{#1}_{#2}} 
 \newcommand{\one}{\mbox{\sc One}}
 \newcommand{\two}{\mbox{\sc Two}}
 \newcommand{\three}{\mbox{\sc Three}}
 \newcommand{\four}{\mbox{\sc Four}}
 \newcommand{\first}{\mbox{\sc Derivation}}
 \newcommand{\second}{\mbox{\sc Second}}
 \newcommand{\uorigin}{\mbox{\sc Org}}
 \newcommand{\image}{\mbox{\sc Img}}
 \newcommand{\limitset}{\mbox{\sc Lim}}
 \newcommand{\fif}{\mbox{\bf CL15}}
\newcommand{\col}[1]{\mbox{$#1$:}}

\newcommand{\sti}{\mbox{\raisebox{-0.02cm}
{\scriptsize $\circ$}\hspace{-0.121cm}\raisebox{0.08cm}{\tiny $.$}\hspace{-0.079cm}\raisebox{0.10cm}
{\tiny $.$}\hspace{-0.079cm}\raisebox{0.12cm}{\tiny $.$}\hspace{-0.085cm}\raisebox{0.14cm}
{\tiny $.$}\hspace{-0.079cm}\raisebox{0.16cm}{\tiny $.$}\hspace{1pt}}}
\newcommand{\costi}{\mbox{\raisebox{0.08cm}
{\scriptsize $\circ$}\hspace{-0.121cm}\raisebox{-0.01cm}{\tiny $.$}\hspace{-0.079cm}\raisebox{0.01cm}
{\tiny $.$}\hspace{-0.079cm}\raisebox{0.03cm}{\tiny $.$}\hspace{-0.085cm}\raisebox{0.05cm}
{\tiny $.$}\hspace{-0.079cm}\raisebox{0.07cm}{\tiny $.$}\hspace{1pt}}}

\newcommand{\seq}[1]{\langle #1 \rangle}           

\newcommand{\pstb}{\mbox{\raisebox{-0.01cm}{\large $\wedge$}\hspace{-5pt}\raisebox{0.26cm}{\small $\mid$}\hspace{4pt}}}
\newcommand{\pcostb}{\mbox{\raisebox{0.22cm}{\large $\vee$}\hspace{-5pt}\raisebox{0.02cm}{\footnotesize $\mid$}\hspace{4pt}}}

\newcommand{\sst}{\mbox{\raisebox{-0.07cm}{\scriptsize $-$}\hspace{-0.2cm}$\pst$}} 

\newcommand{\scost}{\mbox{\raisebox{0.20cm}{\scriptsize $-$}\hspace{-0.2cm}$\pcost$}} 


\newcommand{\mla}{\mbox{{\Large $\wedge$}}}
\newcommand{\mle}{\mbox{{\Large $\vee$}}}

\newcommand{\pst}{\mbox{\raisebox{-0.01cm}{\scriptsize $\wedge$}\hspace{-4pt}\raisebox{0.16cm}{\tiny $\mid$}\hspace{2pt}}}
\newcommand{\gneg}{\neg}                  
\newcommand{\mli}{\rightarrow}                     
\newcommand{\cla}{\mbox{\large $\forall$}}      
\newcommand{\cle}{\mbox{\large $\exists$}}        
\newcommand{\mld}{\vee}    
\newcommand{\mlc}{\wedge}  
\newcommand{\ade}{\mbox{\Large $\sqcup$}}      
\newcommand{\ada}{\mbox{\Large $\sqcap$}}      
\newcommand{\add}{\sqcup}                      
\newcommand{\adc}{\sqcap}                      

\newcommand{\tlg}{\bot}               
\newcommand{\twg}{\top}               
\newcommand{\st}{\mbox{\raisebox{-0.05cm}{$\circ$}\hspace{-0.13cm}\raisebox{0.16cm}{\tiny $\mid$}\hspace{2pt}}}
\newcommand{\cst}{{\mbox{\raisebox{-0.05cm}{$\circ$}\hspace{-0.13cm}\raisebox{0.16cm}{\tiny $\mid$}\hspace{1pt}}}^{\aleph_0}} 
\newcommand{\cost}{\mbox{\raisebox{0.12cm}{$\circ$}\hspace{-0.13cm}\raisebox{0.02cm}{\tiny $\mid$}\hspace{2pt}}}
\newcommand{\ccost}{{\mbox{\raisebox{0.12cm}{$\circ$}\hspace{-0.13cm}\raisebox{0.02cm}{\tiny $\mid$}\hspace{1pt}}}^{\aleph_0}} 
\newcommand{\pcost}{\mbox{\raisebox{0.12cm}{\scriptsize $\vee$}\hspace{-4pt}\raisebox{0.02cm}{\tiny $\mid$}\hspace{2pt}}}

\newcommand{\psti}{\mbox{\raisebox{-0.02cm}{\tiny $\wedge$}\hspace{-0.121cm}\raisebox{0.08cm}{\tiny $.$}\hspace{-0.079cm}\raisebox{0.10cm}
{\tiny $.$}\hspace{-0.079cm}\raisebox{0.12cm}{\tiny $.$}\hspace{-0.085cm}\raisebox{0.14cm}
{\tiny $.$}\hspace{-0.079cm}\raisebox{0.16cm}{\tiny $.$}\hspace{1pt}}}

\newcommand{\pcosti}{\mbox{\raisebox{0.08cm}{\tiny $\vee$}\hspace{-0.121cm}\raisebox{-0.01cm}{\tiny $.$}\hspace{-0.079cm}\raisebox{0.01cm}
{\tiny $.$}\hspace{-0.079cm}\raisebox{0.03cm}{\tiny $.$}\hspace{-0.085cm}\raisebox{0.05cm}
{\tiny $.$}\hspace{-0.079cm}\raisebox{0.07cm}{\tiny $.$}\hspace{1pt}}}


\newtheorem{theoremm}{Theorem}[section]
\newtheorem{conditionss}{Condition}[section]
\newtheorem{thesiss}[theoremm]{Thesis}
\newtheorem{definitionn}[theoremm]{Definition}
\newtheorem{lemmaa}[theoremm]{Lemma}
\newtheorem{notationn}[theoremm]{Notation}\newtheorem{corollary}[theoremm]{Corollary}
\newtheorem{propositionn}[theoremm]{Proposition}
\newtheorem{conventionn}[theoremm]{Convention}
\newtheorem{examplee}[theoremm]{Example}
\newtheorem{remarkk}[theoremm]{Remark}
\newtheorem{factt}[theoremm]{Fact}
\newtheorem{exercisee}[theoremm]{Exercise}
\newtheorem{questionn}[theoremm]{Open Problem}
\newtheorem{conjecturee}[theoremm]{Conjecture}

\newenvironment{exercise}{\begin{exercisee} \em}{ \end{exercisee}}
\newenvironment{definition}{\begin{definitionn} \em}{ \end{definitionn}}
\newenvironment{theorem}{\begin{theoremm}}{\end{theoremm}}
\newenvironment{lemma}{\begin{lemmaa}}{\end{lemmaa}}
\newenvironment{proposition}{\begin{propositionn} }{\end{propositionn}}
\newenvironment{convention}{\begin{conventionn} \em}{\end{conventionn}}
\newenvironment{remark}{\begin{remarkk} \em}{\end{remarkk}}
\newenvironment{proof}{ {\bf Proof.} }{\  \rule{2.5mm}{2.5mm} \vspace{.2in} }
\newenvironment{idea}{ {\bf Idea.} }{\  \rule{1.5mm}{1.5mm} \vspace{.15in} }
\newenvironment{example}{\begin{examplee} \em}{\end{examplee}}
\newenvironment{fact}{\begin{factt}}{\end{factt}}
\newenvironment{notation}{\begin{notationn} \em}{\end{notationn}}
\newenvironment{conditions}{\begin{conditionss} \em}{\end{conditionss}}
\newenvironment{question}{\begin{questionn}}{\end{questionn}}
\newenvironment{conjecture}{\begin{conjecturee}}{\end{conjecturee}}

\title{A cirquent calculus system with clustering and ranking\thanks{Supported by National Natural Science Foundation of China (61303030) and the Fundamental Research Funds for the Central Universities of China (K5051370023).}}
\author{Wenyan Xu\\
{\it\small School of Mathematics and Statistics, Xidian University, Xi'an 710071, China}}
\date{}
\maketitle

\begin{abstract}
Cirquent calculus is a new proof-theoretic and semantic approach introduced by G.Japaridze for the needs of his theory of computability logic. The earlier article ``From formulas to cirquents in computability logic" by Japaridze generalized the concept of cirquents to the version with what are termed {\it clusterng} and {\it ranking}, and showed that, through cirquents with clustering and ranking, one can capture, refine and generalize the so called {\em extended IF logic}. Japaridze's treatment of extended IF logic, however, was purely semantical, and no deductive system was proposed. The present paper syntactically constructs a cirquent calculus system with clustering and ranking, sound and complete w.r.t. the propositional fragment of cirquent-based semantics. Such a system can be considered not only a conservative extension of classical propositional logic but also, when limited to cirquents with $\leq 2$ ranks, an axiomatization of purely propositional extended IF logic in its full generality.

\end{abstract}


\



\section{Introduction}\label{ssintr}
{\em Cirquent calculus} is a new proof-theoretic and semantic approach introduced by G.Japaridze \cite{Cir} for the needs of his theory of computability logic \cite{Jap03,beginning}.  Its main characteristic feature is being based on circuit-style structures called {\em cirquents}, as opposed to the more traditional approaches that manipulate tree-like objects such as formulas. Cirquents, unlike formulas, are allowing (one or another sort of) {\em sharing} of subcomponents between different components. Due to sharing, cirquent calculus has greater expressiveness and higher efficiency. For instance, as shown in \cite{deep}, the analytic cirquent calculus system CL8 achieves an exponential speedup of proofs over the classical analytic systems. Since its birth, cirquent calculus has been developed in a series of articles \cite{deep,fromto,taming1,taming2,wenyan1,wenyan2,wenyan3,wenyan4}.

The concept of cirquents was qualitatively generalized in \cite{fromto}, where the ideas of {\em clustering} and {\em ranking} were introduced. Intuitively, clusters are generalized disjunctive or conjunctive gates, i.e. switch-style devices that combine tuples of individual gates of a given type in a parallel way ---  in a way where the  choice ({\em left} or {\em right}) of an argument is shared between all members. Ranks are superior consoles of a certain subset of clusters, with all such consoles arranged in a linear order indicating in what order selections by the consoles should be made.

It was showed semantically in \cite{fromto} that, through cirquents with clustering both disjunctions and conjunctions and ranking, one can capture, refine and generalize the conservative extension of {\em independence-friendly (IF) logic} \cite{IF,IFBOOK} known as {\em extended IF logic} (cf. \cite{tul}). The latter, in addition to what IF logic calls {\em strong negation $\sim$}, also considers {\em weak negation $\neg$}.
The main distinguishing feature of (extended) IF logic is allowing one to express independence relations between quantifiers. But the past attempts (cf. \cite{A.P.,S.andP.}) to develop (extended) IF logic at the purely propositional level have remained limited to some special syntactic fragments of the language. Then the approach in \cite{fromto} allows
one to account for independence from propositional connectives in the same spirit as 
(extended) IF logic accounts for independence from quantifiers.

Japaridze's treatment of extended IF logic in \cite{fromto}, however, was purely semantical, and no deductive system was proposed. In this paper, we axiomatically construct a cirquent calculus system, called {\bf $RIF_p$}, with clustered disjunctive and conjunctive connectives and $n$ ranks for any positive integer $n$. Such a system is proved to be sound and complete w.r.t. the propositional fragment of Japaridze's cirquent-based semantics, and can be considered not only a conservative extension of classical propositional logic but also, when limited to cirquents with $\leq 2$ ranks, an axiomatization of purely propositional extended IF logic in its full generality.

\section{Preliminaries}
In this section we reproduce the basic concepts from \cite{fromto} on which the later parts of the paper will rely.
An interested reader may consult \cite{fromto} for additional
explanations, illustrations and examples.

Our propositional language has infinitely many  {\bf atoms}, for which $p,q,r,s,\ldots$ will be used as metavariables. An atom $p$ and its negation $\neg p$ are called {\bf literals}.
A {\bf formula} means one of the language of classical propositional logic, built from
literals and the binary connectives $\wedge,\vee$ in the standard way.  $A\rightarrow B$ is understood as an abbreviation of $\neg A\vee B$. And $\neg$, when applied to anything other than an atom, is  understood as an abbreviation defined by $\neg\neg A=A$, $\neg(A\wedge B)=\neg A\vee\neg B$ and $\neg(A\vee B)=\neg A\wedge\neg B$. Namely, all formulas are required to be in negation normal form.

A {\bf cirquent} is a formula together with two extra parameters called {\bf clustering} and {\bf ranking}, respectively. Clustering is a partition of the set of all occurrences of $\vee,\wedge$ into subsets, called {\bf clusters}, satisfying the condition that all occurrences of $\vee,\wedge$ within any given cluster have the same type. Ranking is a partition of the set of all $\vee$ and $\wedge$ clusters into subsets, called {\bf ranks}, arranged in a linear order, with each rank satisfying the condition that all clusters in it have the same type.\footnote{The concept of cirquents considered in cirquent calculus is more general than the one defined here. See \cite{fromto}.}  Each cluster is associated with a unique positive integer called its {\bf ID}. IDs serve as identifiers for clusters, and we will simply say  ``cluster $k$" to mean ``the cluster whose ID is $k$''. A rank containing $\wedge$-clusters is said to be {\bf conjunctive}, and a rank containing $\vee$-clusters {\bf disjunctive}. Since the ranks are linearly ordered, we will refer to them as the 1st rank, the 2nd rank, etc. or rank 1, rank 2, etc. Also, instead of ``cluster $k$ is in the $i$th  rank", we will say ``$k$ is of rank $i$".

One way to represent cirquents is to do so graphically, using arcs to indicate the occurrences of the connectives' ``clusteral affiliations" and  the clusters' ``rankal affiliations"  as  in the following figure:
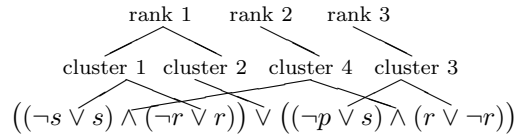
\begin{center}
\begin{picture}(20,50)(0,0)
\put(-90,0){$\bigl((\neg s\vee s)\wedge(\neg r\vee r)\bigr)\vee\bigl((\neg p\vee s)\wedge(r\vee \neg r)\bigr)$}
\put(-46,17){\line(-5,-3){18}}\put(-46,17){\line(5,-2){28}}\put(-70,19){\footnotesize cluster $1$}
\put(24,17){\line(-6,-1){68}}\put(24,17){\line(3,-1){32}}\put(8,19){\footnotesize cluster $4$}
\put(58,17){\line(-2,-1){20}}\put(58,17){\line(2,-1){19}}\put(48,19){\footnotesize cluster $3$}
\put(-32,19){\footnotesize cluster $2$}\put(-20,17){\line(5,-2){26}}

\put(-45,39){\footnotesize rank $1$}\put(-7,39){\footnotesize rank $2$}\put(30,39){\footnotesize rank $3$}
\put(-32,37){\line(-2,-1){22}}\put(-32,37){\line(2,-1){22}}
\put(5,37){\line(2,-1){22}}
\put(40,37){\line(2,-1){22}}
\end{picture}
\end{center}
For space efficiency considerations, in this paper we will instead be writing cirquents just like formulas, only  with every occurrence of $\vee$ (resp. $\wedge$) indexed with a symbol $k^{i}$, called the {\bf index} of this occurrence, as an indication that this occurrence belongs to cluster $k$, and that such a cluster $k$ is of rank $i$.  So, for instance, the above cirquent will be simply written as
$\bigl((\neg s\vee_{1^{1}}s)\wedge_{4^{2}}(\neg r\vee_{1^{1}}r)\bigr)\vee_{2^{1}}\bigl((\neg p\vee_{3^{3}} s)\wedge_{4^{2}}(r\vee_{3^{3}} \neg r)\bigr)$.

A cirquent $\mathcal{C}$ is said to be {\bf classical} iff all of its clusters are singletons.
We shall identify such a cirquent with the formula of classical  logic obtained from it by simply deleting all indices, i.e. replacing each $\vee_{k^{i}}$ (resp. $\wedge_{k^i}$)(whatever $k$ and $i$) with just $\vee$ (resp. $\wedge$). Throughout the rest of this paper, we will be using the term {\bf oconnective} to
refer to a connective together with a particular occurrence of it in a cirquent.

An {\bf interpretation} (or {\bf model}) is a function $^{\ast}$ that sends each atom $p$ to one of the values $p^{\ast}\in\{\top,\bot\}$, and extends to all literals by stipulating that $(\neg p)^{\ast}=\top$ iff $p^{\ast}=\bot$.

Let $i\in\{1,2,3,\ldots\}$. An {\bf $i$-metaselection} is a function $f_i:\{1,2,3,\ldots\}\rightarrow\{\mbox{\em left,right}\}$. Let $\mathcal{C}$ be a cirquent with $n$ ranks.
A {\bf metaselection} for $\mathcal{C}$ is an $n$-tuple $\overrightarrow{f}=(f_1,\ldots,f_n)$ where, for each $1\leq i\leq n$, $f_i$ is an $i$-metaselection.

Given a cirquent $\mathcal{C}$ with $n$ ranks and a metaselection $\overrightarrow{f}=(f_1,\ldots,f_n)$ for $\mathcal{C}$, the {\bf resolvent} of a disjunctive (resp. conjunctive) subcirquent $\mathcal{A}\vee_{k^{i}} \mathcal{B}$ (resp. $\mathcal{A}\wedge_{k^{i}} \mathcal{B}$) of  $\mathcal{C}$ is defined to be $\mathcal{A}$ if $f_i(k)=\mbox{\em left}$, and $\mathcal{B}$ if $f_i(k)=\mbox{\em right}$.

Let $\mathcal{C}$ be a cirquent with $n$ ranks, $^{\ast}$ an interpretation, and $\overrightarrow{f}=(f_1,\ldots,f_n)$ a metaselection for $\mathcal{C}$. In this context, with ``metatrue" to be read as ``{\bf metatrue w.r.t. $(^{\ast},\overrightarrow{f})$}", we say that:
\begin{itemize}
\item A literal $L$ of $\mathcal{C}$ is metatrue iff $L^{\ast}=\top$.
\item A subcirquent $\mathcal{A}\vee_{k^{i}}\mathcal{B}$ (resp. $\mathcal{A}\wedge_{k^{i}}\mathcal{B}$) of $\mathcal{C}$  is metatrue iff so is its resolvent.
\end{itemize}
Next, we say that $\mathcal{C}$ is {\bf true} under the interpretation $^*$ (in the model $^*$), or simply that $\mathcal{C}^*$ is true, iff
\begin{center}
$\mathcal{Q}_1f_1\ldots\mathcal{Q}_nf_n$ such that $\mathcal{C}$ is metatrue w.r.t. $(^{\ast},(f_1,\ldots,f_n))$.
\end{center}
Here each $\mathcal{Q}_if_i$ ($1\leq i\leq n$) abbreviates the phrase ``for every $i$-metaselection $f_i$" if the $i$th rank of $\mathcal{C}$ is conjunctive, and ``there is an $i$-metaselection $f_i$" if the $i$th rank of $\mathcal{C}$ is disjunctive.
Finally, we say that $\mathcal{C}$ is {\bf valid} iff it is true under every interpretation (in every model).

Note that, when $\mathcal{C}$ is a formula, i.e. a cirquent where all clusters are singletons, $\mathcal{C}$ is valid iff it is valid (tautological) in the sense of classical logic. And classical truth of a formula under an interpretation $^*$ means nothing but $\mathcal{Q}_1f_1\ldots\mathcal{Q}_nf_n$ such that the formula is metatrue in our sense w.r.t. $(^{\ast},(f_1,\ldots,f_n))$. So, classical logic is nothing but the conservative fragment of our logic obtained by only allowing formulas in the language.

The rest of this section is not technically relevant to the main results of the present paper and is
only meant to understand
what all of the above has to do with extended IF logic.\footnote{For those unfamiliar with (extended) IF logic, they may want to consult \cite{IFBOOK} for the basic IF logic or \cite{tul} for the extended one.} Consider the  formula
\begin{equation}\label{wenzi0}
\forall x(\exists y/\forall x)\hspace{2pt} p(x,y)
\end{equation}
with its standard meaning. According to the latter, given any object $x$, the object $y$ can be chosen so that $p(x,y)$ is true, with the modifier `$\forall x$' attached to `$\exists y$' indicating that here $y$ can be chosen independently from (without any knowledge of)  $x$. Assuming that the universe of discourse is $\{1,2\}$, (\ref{wenzi0}) can just as well be (re-)written as
\begin{equation}
\bigl(p(1,1)\vee^y\hspace{-3pt}/\hspace{-3pt}\wedge^x p(1,2)\bigr)\wedge^x\bigl(p(2,1)\vee^y\hspace{-3pt}/\hspace{-3pt}\wedge^x p(2,2)\bigr),
\end{equation}
which, after further rewriting $p(1,1),p(1,2),p(2,1),p(2,2)$ as the more compact $p,q,r,s$, is the propositional formula
\begin{equation}\label{wenzi}
(p\vee^y\hspace{-3pt}/\hspace{-3pt}\wedge^x q)\wedge^x(r\vee^y\hspace{-3pt}/\hspace{-3pt}\wedge^x s).
\end{equation}
Here we have turned $\forall x$ into $\wedge^x$, $\exists y$ into $\vee^y$, with the superscript in each case used just to remind us from which quantifier each oconnective was obtained, and $\vee^y\hspace{-1pt}/\hspace{-1pt}\wedge^x$ indicating that the $y$-superscripted disjunction is independent of the $x$-superscripted conjunction. Now, Japaridze's recipes (see \cite{fromto}, Descriptions 7.4, 7.5 and 7.6) translate (\ref{wenzi}) into the following cirquent:
\begin{equation}\label{wenzi2}
(p\vee_{1^1} q)\wedge_{2^2}(r\vee_{1^1} s).
\end{equation}
Note that cluster $1$ contains two disjunctive oconnectives --- namely, those originating from $\exists y$,  cluster $2$ is a conjunctive singleton, all disjunctive clusters are of rank 1, and all conjunctive clusters are of rank 2. It is left as an exercise for the reader to convince himself or herself that, in any given model (interpretation) $^*$, (\ref{wenzi2}) is true in our sense if and only if (\ref{wenzi}) is true in the sense of extended IF logic.
Similarly, the two forms of (\ref{wenzi0})'s negation $\neg\forall x(\exists y/\forall x)\hspace{2pt} p(x,y)$ and $\sim\forall x(\exists y/\forall x)\hspace{2pt} p(x,y)$ are translated into the cirquents $(\neg p\wedge_{1^1} \neg q)\vee_{2^2}(\neg r\wedge_{1^1} \neg s)$ and $(\neg p\wedge_{1^2} \neg q)\vee_{2^1}(\neg r\wedge_{1^2} \neg s)$, respectively. Since the above translations only generate cirquents with $\leq 2$  ranks, these sorts of cirquents
are sufficient for capturing extended IF logic.

Well, the present case is a ``lucky'' case because we easily understand what ``true in the sense of extended IF logic'' means for (\ref{wenzi}) --- after all, (\ref{wenzi}) originates from (and will be handled in the same way as) the first-order (\ref{wenzi0}). As an example of an ``unlucky'' case, consider the cirquent
\begin{equation} \label{wenzi3}
(r\vee_{1^1} s)\wedge_{2^2}\bigl((p\vee_{1^1} q)\wedge_{3^2} q\bigr).
 \end{equation}
It is just as meaningful from the point of view of our semantics as any other cirquent, including (\ref{wenzi2}). An attempt to express the same in the traditional formalism of IF logic apparently yields something like
\begin{equation} \label{wenzi33}
(r\vee\hspace{-3pt}/\hspace{-3pt}\wedge^x s)\wedge^x\bigl((p\vee\hspace{-3pt}/\hspace{-3pt}\wedge^x q)\wedge^y q\bigr).
\end{equation}
Unlike (\ref{wenzi}), however, (\ref{wenzi33}) is problematic for the traditional semantical approaches (the ones based on imperfect information games) to IF logic. Namely, because of a problem called signaling, it is far from clear how its truth should be understood. 

If the connections and differences between our present semantics and that of extended IF logic are still not clear, see the first 7 sections of \cite{fromto} for more explanations, discussions and examples.

\section{Main results}
\subsection{System $RIF_p$ introduced}
Our system introduced in this paper is called {\bf $RIF_p$} (ranked IF logic at the propositional level). As will be seen shortly, the inference rules of $RIF_p$
modify cirquents at any level rather than only around the root. Thus, $RIF_p$ is in fact a {\em deep inference} system, in the style of \cite{deepinference}. This explains our borrowing some notation from the Calculus of Structures.
Namely, we
will be using $\Phi\{\}$ or $\Psi\{\}$ to denote any cirquent where a vacancy (``hole'') $\{\}$ appears in the place of a subcirquent. The vacancy $\{\}$ can be filled with any cirquent. For example, if $\Phi\{\} = (p\vee_{1^{2}} q)\vee_{1^{2}}(\{\}\wedge_{2^{1}}q)$, then $\Phi\{\neg p\} = (p\vee_{1^{2}} q)\vee_{1^{2}}(\neg p\wedge_{2^{1}} q)$, $\Phi\{q\} = (p\vee_{1^{2}} q)\vee_{1^{2}}(q\wedge_{2^{1}} q)$, and $\Phi\{p\vee_{1^{2}} q\} = (p\vee_{1^{2}} q)\vee_{1^{2}}((p\vee_{1^{2}} q)\wedge_{2^{1}} q)$.

Further, we will be using $\mathcal{C}[k^i]$ to denote a cirquent $\mathcal{C}$ that contains some occurrence of $k^i$ in its representation, where $k^i$ is the index of some oconnective of $\mathcal{C}$. And we will be using $\mathcal{C}[k^i/l^i]$ to denote the resulting cirquent from $\mathcal{C}[k^i]$ by replacing all the occurrences of $k^i$ in $\mathcal{C}[k^i]$ by $l^i$.

Below comes the {\bf inference rules} of $RIF_p$, where $\mathcal{A,B,C,D}$ stand for any cirquents; $\odot$ and $\circ$ are variables over $\{\wedge,\vee\}$.\footnote{In the remaining of this paper, without any further indication, $\odot$ (resp. $\circ$) will always stand for a variable over $\{\wedge,\vee\}$.} It is important to point out that,  in each rule, {\em all} occurrences of $\odot$ (resp. $\circ$) stand for the same object.\vspace{2mm}

{\bf Rule I:} This rule has two versions, Rule I (left) and Rule I (right), as shown in the following figure, where $k^{i}$ are indices of oconnectives $\odot$.
\begin{center}
\begin{picture}(80,40)(0,165)
\put(-45,0){\begin{picture}(80,200)\put(-79,190){$\Phi\bigl\{\Psi\{\mathcal{A}\}\odot_{k^{i}}\mathcal{C}\bigr\}$}
\put(-89,184){\line(1,0){93}}\put(7,183){\footnotesize\bf Rule I (left)}
\put(-90,172){$\Phi\bigl\{\Psi\{\mathcal{A}\odot_{k^{i}}\mathcal{B}\}\odot_{k^{i}}\mathcal{C}\bigr\}$}\end{picture}}

\put(150,0){\begin{picture}(80,200)\put(-79,190){$\Phi\bigl\{\mathcal{C}\odot_{k^{i}}\Psi\{\mathcal{A}\}\bigr\}$}
\put(-89,184){\line(1,0){93}}\put(7,183){\footnotesize\bf Rule I (right)}
\put(-90,172){$\Phi\bigl\{\mathcal{C}\odot_{k^{i}}\Psi\{\mathcal{B}\odot_{k^{i}}\mathcal{A}\}\bigr\}$}\end{picture}}
\end{picture}
\end{center}

{\bf Rule II:} This rule also has two versions, Rule II (left) and Rule II (right), as shown in the following figure, where (i) $k^{i},l^{j},m^{j},n^{j}$ are indices of oconnectives satisfying the condition that when cluster $l$ is a non-singleton in the conclusion, $m=n=l$; when cluster $l$ is a singleton in the conclusion, $m,n$ are any positive integers such that clusters $m,n$ are sigletons in the premise; (ii) $i\leq j$; (iii) all the ranks $t$ that occurs in $\mathcal{C}$ (in the conclusion) satisfy the condition that $t\geq i$; (iv) $\mathcal{C}_1$ (resp. $\mathcal{C}_2$) is the resulting cirquent from $\mathcal{C}$  by replacing in $\mathcal{C}$ each singleton cluster $s$ of (whatever) rank $r$ with singleton cluster $s_1$ (resp. $s_2$) of the same rank $r$ in the premise.
\begin{center}
\begin{picture}(80,40)(0,118)
\put(-63,-50){\begin{picture}(80,200)\put(-95,190){$\Phi\bigl\{(\mathcal{A}\circ_{m^{j}}\mathcal{C}_1)\odot_{k^{i}}(\mathcal{B}\circ_{n^{j}}\mathcal{C}_2)\bigr\}$}
\put(-95,184){\line(1,0){125}}\put(34,183){\footnotesize\bf Rule II (left)}
\put(-83,172){$\Phi\bigl\{(\mathcal{A}\odot_{k^{i}}\mathcal{B})\circ_{l^{j}}\mathcal{C}\bigr\}$}\end{picture}}

\put(152,-50){\begin{picture}(80,200)\put(-93,190){$\Phi\bigl\{(\mathcal{C}_1\circ_{m^{j}}\mathcal{A})\odot_{k^{i}}(\mathcal{C}_2\circ_{n^{j}}\mathcal{B})\bigr\}$}
\put(-93,184){\line(1,0){125}}\put(35,183){\footnotesize\bf Rule II (right)}
\put(-82,172){$\Phi\bigl\{\mathcal{C}\circ_{l^{j}}(\mathcal{A}\odot_{k^{i}}\mathcal{B})\bigr\}$}\end{picture}}
\end{picture}
\end{center}

{\bf Rule III:} This rule is shown in the following figure, where (i) $k^{i},l^{j},m^{j},n^{j}$ are indices of oconnectives satisfying the condition that when cluster $l$ is a non-singleton in the conclusion, $m=n=l$; when cluster $l$ is a singleton in the conclusion, $m,n$ are any positive integers such that clusters $m,n$ are sigletons in the premise; (ii) $i\leq j$.
\begin{center}
\begin{picture}(80,40)(-110,68)
\put(-53,-100){\begin{picture}(80,200)\put(-95,190){$\Phi\bigl\{(\mathcal{A}\circ_{m^{j}}\mathcal{C})\odot_{k^{i}}(\mathcal{B}\circ_{n^{j}}\mathcal{D})\bigr\}$}
\put(-97,184){\line(1,0){124}}\put(30,183){\footnotesize\bf Rule III}
\put(-94,172){$\Phi\bigl\{(\mathcal{A}\odot_{k^{i}}\mathcal{B})\circ_{l^{j}}(\mathcal{C}\odot_{k^{i}}\mathcal{D})\bigr\}$}\end{picture}}
\end{picture}
\end{center}

{\bf Rule IV:} This rule is shown in the following figure, where (i) $k^{i},l^{j},r^{j}$ are indices of oconnectives satisfying the conditions that $l$ only occurs in the subcirquents $\mathcal{A}$ and $\mathcal{B}$ in the conclusion, and that $r$ (in the premise) is any positive integer such that $r$ does not occur in the conclusion; (ii) $i<j$.
\begin{center}
\begin{picture}(80,40)(90,68)
\put(142,-100){\begin{picture}(80,200)
\put(-75,190){$\Phi\bigl\{\mathcal{A}[l^{j}]\odot_{k^{i}}\mathcal{B}[l^{j}/r^{j}]\bigr\}$}
\put(-80,184){\line(1,0){100}}\put(23,183){\footnotesize\bf Rule IV}
\put(-68,172){$\Phi\bigl\{\mathcal{A}[l^{j}]\odot_{k^{i}}\mathcal{B}[l^{j}]\bigr\}$}
\end{picture}}
\end{picture}
\end{center}

It is obvious that all rules of $RIF_p$ preserve the liner order of ranks of cirquents in both top-down and bottom-up directions. In each application of these rules, we call the oconnective(s) $\odot_{k^{i}}$ in the premise (resp. conclusion), as shown in the above corresponding figures, the {\bf key oconnecitve(s)} of this application in the premise (resp. conclusion).\vspace{2mm}

The {\bf axioms} of  $RIF_p$ are all classical cirquents that (seen as formulas)
are tautologies of classical propositional logic.\vspace{2mm}

A {\bf proof} of a cirquent $\mathcal{C}$ in $RIF_p$ is a sequence of cirquents such that the first cirquent in the sequence is an axiom of $RIF_p$, the last cirquent is  $\mathcal{C}$, and every cirquent, except the axiom, follows from the preceding cirquent by one of the rules of $RIF_p$. When such a proof exists, $\mathcal{C}$ is said to be {\bf provable} in $RIF_p$.

\begin{lemma}\label{yan2}
Given a cirquent $\Phi\{\mathcal{A}\odot_{k^{i}}\mathcal{B}\}$, a metaselection $\overrightarrow{f}=(f_1,\ldots,f_n)$ for it and an  interpretation $^{\ast}$, the cirquent $\Phi\{\mathcal{A}\odot_{k^{i}}\mathcal{B}\}$ is metatrue w.r.t. $(^{\ast},\overrightarrow{f})$ iff $f_i(k)=\mbox{\em left}$ (resp. $f_i(k)=\mbox{\em right}$)  and
the cirquent $\Phi\{\mathcal{A}\}$ (resp. $\Phi\{\mathcal{B}\}$) is metatrue w.r.t. $(^{\ast},\overrightarrow{f})$.
\end{lemma}
\begin{proof}
We prove the proposition by induction on the number of oconnectives of $\Phi\{\}$.

For the basis, assume that the number of oconnectives of $\Phi\{\}$ is $0$. Then $\Phi\{\mathcal{A}\odot_{k^{i}}\mathcal{B}\}=\mathcal{A}\odot_{k^{i}}\mathcal{B}$, $\Phi\{\mathcal{A}\}=\mathcal{A}$ and $\Phi\{\mathcal{B}\}=\mathcal{B}$. By the definition of metatruth, we immediately have $\mathcal{A}\odot_{k^{i}}\mathcal{B}$ is metatrue w.r.t. $(^{\ast},\overrightarrow{f})$ if and only if $f_i(k)=\mbox{\em left}$ (resp. $f_i(k)=\mbox{\em right}$) and its resolvent $\mathcal{A}$ (resp. $\mathcal{B}$) is so.

Now (induction hypothesis) assume that the proposition holds when the number of oconnectives of $\Phi\{\}$ is $n$.
We want to show that the proposition still holds when the number of oconnectives of $\Phi\{\}$ is $n+1$. Two cases (i), (ii) are to be considered here:

(i) Assume that the main connective of $\Phi\{\}$ is $\circ\in\{\wedge,\vee\}$ whose index is $l^{j}$ such that cluster $l$ is a singleton. Namely, assume $\Phi\{\mathcal{A}\odot_{k^{i}}\mathcal{B}\}=\Psi\{\mathcal{A}\odot_{k^{i}}\mathcal{B}\}\circ_{l^{j}}\mathcal{C}$  for some cirquent $\mathcal{C}$ (the other possibility
$\Phi\{\mathcal{A}\odot_{k^{i}}\mathcal{B}\}= \mathcal{C} \circ_{l^{j}}\Psi\{\mathcal{A}\odot_{k^{i}}\mathcal{B}\}$ is similar).  The following two (sub)cases need to be further considered.

{\it Case (a)}: $f_j(l)=\mbox{\em right}$.  Then $\Psi\{\mathcal{A}\odot_{k^{i}}\mathcal{B}\}\circ_{l^{j}}\mathcal{C}$ is metatrue w.r.t. $(^{\ast},\overrightarrow{f})$ iff its resolvent $\mathcal{C}$ is so. But exactly the same holds for both $\Psi\{\mathcal{A}\}\circ_{l^{j}}\mathcal{C}$ and $\Psi\{\mathcal{B}\}\circ_{l^{j}}\mathcal{C}$ (for the same reason).  Thus, vacuously adding ``$f_i(k)=\ldots$'', we arrive at the desired conclusion that $\Psi\{\mathcal{A}\odot_{k^{i}}\mathcal{B}\}\circ_{l^{j}}\mathcal{C}$ is metatrue w.r.t. $(^{\ast},\overrightarrow{f})$ iff $f_i(k)=\mbox{\em left}$ (resp. $f_i(k)=\mbox{\em right}$) and $\Psi\{\mathcal{A}\}\circ_{l^{j}}\mathcal{C}$ (resp. $\Psi\{\mathcal{B}\}\circ_{l^{j}}\mathcal{C}$) is metatrue w.r.t. $(^{\ast},\overrightarrow{f})$.

{\it Case (b)}: $f_j(l)=\mbox{\em left}$. Then $\Psi\{\mathcal{A}\odot_{k^{i}}\mathcal{B}\}\circ_{l^{j}}\mathcal{C}$ is metatrue w.r.t. $(^{\ast},\overrightarrow{f})$ iff its resolvent $\Psi\{\mathcal{A}\odot_{k^{i}}\mathcal{B}\}$ is so, which, in turn (by the induction hypothesis), is the case  iff  $f_i(k)=\mbox{\em left}$ (resp. $f_i(k)=\mbox{\em right}$) and $\Psi\{\mathcal{A}\}$ (resp. $\Psi\{\mathcal{B}\}$) is metatrue w.r.t. $(^{\ast},\overrightarrow{f})$. This, in turn, is the case iff $f_i(k)=\mbox{\em left}$ (resp. $f_i(k)=\mbox{\em right}$) and $\Phi\{\mathcal{A}\}$ (resp. $\Phi\{\mathcal{B}\}$) is metatrue w.r.t. $(^{\ast},\overrightarrow{f})$. Hence the desired conclusion holds.

(ii) Assume that the main connective of $\Phi\{\}$ is $\circ\in\{\wedge,\vee\}$ whose index is $m^{j}$ such that cluster $m$ is a non-singleton. Namely, assume $\Phi\{\mathcal{A}\odot_{k^{i}}\mathcal{B}\}=\Psi\{\mathcal{A}\odot_{k^{i}}\mathcal{B}\}\circ_{m^{j}}\mathcal{C}$  for some cirquent $\mathcal{C}$ (the other possibility
$\Phi\{\mathcal{A}\odot_{k^{i}}\mathcal{B}\}= \mathcal{C} \circ_{m^{j}}\Psi\{\mathcal{A}\odot_{k^{i}}\mathcal{B}\}$ is similar). If $m\neq k$, then we can employ an essentially the same argument as the one used in (i).  And if
$m=k$ (so that $i=j$), then the case is even simpler, so we leave details to the reader.
\end{proof}

\begin{lemma}\label{niu}
All rules of $RIF_p$ preserve truth in both top-down and bottom-up directions.
\end{lemma}
\begin{proof}
Pick an arbitrary interpretation $^*$.\vspace{2mm}

{\em Rule I}: Here we only consider Rule I (left),  with Rule I (right) being similar. We want to show that the premise $\Phi\bigl\{\Psi\{\mathcal{A}\}\odot_{k^{i}}\mathcal{C}\bigr\}$ is true under $^{*}$ iff so is the conclusion $\Phi\bigl\{\Psi\{\mathcal{A}\odot_{k^{i}}\mathcal{B}\}\odot_{k^{i}}\mathcal{C}\bigr\}$.

(i) Suppose that the conclusion $\Phi\bigl\{\Psi\{\mathcal{A}\odot_{k^{i}}\mathcal{B}\}\odot_{k^{i}}\mathcal{C}\bigr\}$ is true under $^{*}$. Let $1,\ldots,n$ be the linear order of all ranks of the conclusion, and $n_1,\ldots,n_m$ ($m\leq n$) that of the premise. Obviously, $n_1,\ldots,n_m$ is a subsequence of $1,\ldots,n$.

For any $j$ ($1\leq j\leq m$), let $\mathcal{O}_{n_j}g_{n_j}$ abbreviate the phrase {\scriptsize $(1)$} ``for every $n_j$-metaselection $g_{n_j}$" if the $n_j$th rank of the premise is conjunctive, and the phrase {\scriptsize $(2)$} ``let $g_{n_j}$ be the $n_j$-metaselection $f_{n_j}$ which is the one in phrase {\scriptsize $(4)$}" if the $n_j$th rank of the premise is disjunctive. Let $\mathcal{P}_{n_j}f_{n_j}$ abbreviate the phrase {\scriptsize $(3)$} `` for the $n_j$-metaselection $f_{n_j}=g_{n_j}$ where $g_{n_j}$ comes from phrase {\scriptsize $(1)$}" if the $n_j$th rank of the conclusion is conjunctive, and the phrase {\scriptsize $(4)$} ``there is a $n_j$-metaselection $f_{n_j}$" if the $n_j$th rank of the conclusion is disjunctive.

Then, by the definition of truth (of the conclusion), we have that $\mathcal{Q}_1f_1\ldots\mathcal{Q}_nf_n$ such that the conclusion $\Phi\bigl\{\Psi\{\mathcal{A}\odot_{k^{i}}\mathcal{B}\}\odot_{k^{i}}\mathcal{C}\bigr\}$ is metatrue w.r.t. $(^{*},(f_1,\ldots,f_n))$, where, for any $r\in\{1,\ldots,n\}$, $\mathcal{Q}_rf_r$ satisfies the conditions that: (a) when $r=n_j$ for some $j\in\{1,\ldots,m\}$, $\mathcal{Q}_rf_r=\mathcal{P}_rf_r$; (b) otherwise $\mathcal{Q}_rf_r$ abbreviates the phrase ``for every $r$-metaselection $f_r$" if the $r$th rank of the conclusion is conjunctive, and ``there is a $r$-metaselection $f_r$" if the $r$th rank of the conclusion is disjunctive.

By lemma \ref{yan2} (applied twice), $\Phi\bigl\{\Psi\{\mathcal{A}\odot_{k^{i}}\mathcal{B}\}\odot_{k^{i}}\mathcal{C}\bigr\}$ is metatrue w.r.t. $(^{\ast},(f_1,\ldots,f_n))$ iff $f_i(k)=\mbox{\em left}$ (resp. $f_i(k)=\mbox{\em right}$)  and
$\Phi\bigl\{\Psi\{\mathcal{A}\}\bigr\}$ (resp. $\Phi\{\mathcal{C}\}$) is metatrue w.r.t. $(^{\ast},(f_1,\ldots,f_n))$. This, in turn, is the case iff $g_i(k)=\mbox{\em left}$ (resp. $g_i(k)=\mbox{\em right}$)  and
$\Phi\bigl\{\Psi\{\mathcal{A}\}\bigr\}$ (resp. $\Phi\{\mathcal{C}\}$) is metatrue w.r.t. $(^{\ast},(g_{n_1},\ldots,g_{n_m}))$, which, in turn, is the case iff the premise $\Phi\bigl\{\Psi\{\mathcal{A}\}\odot_{k^{i}}\mathcal{C}\bigr\}$ is metatrue w.r.t. $(^{\ast},(g_{n_1},\ldots,g_{n_m}))$. Finally, we have $\mathcal{O}_{n_1}g_{n_1}\ldots\mathcal{O}_{n_m}g_{n_m}$ such that the premise $\Phi\bigl\{\Psi\{\mathcal{A}\}\odot_{k^{i}}\mathcal{C}\bigr\}$ is metatrue w.r.t. $(^{\ast},(g_{n_1},\ldots,g_{n_m}))$, which means the premise is true under $^{*}$.

(ii) Now suppose that the premise $\Phi\bigl\{\Psi\{\mathcal{A}\}\odot_{k^{i}}\mathcal{C}\bigr\}$ is true under $^{*}$. Let $1,\ldots,m$ be the linear order of all ranks of the premise, and $1,\ldots,n$ ($n\geq m$) that of the conclusion.

For any $j$ ($1\leq j\leq m$), let $\mathcal{P}_jf_j$ abbreviate the phrase {\scriptsize $(1)$} ``for every $j$-metaselection $f_j$" if the $j$th rank of the conclusion is conjunctive, and the phrase {\scriptsize $(2)$} ``let $f_j$ be the $j$-metaselection $g_j$ which is the one in phrase {\scriptsize $(4)$}" if the $j$th rank of the conclusion is disjunctive. Let $\mathcal{O}_jg_j$ abbreviate the phrase {\scriptsize $(3)$} ``for the $j$-metaselection $g_j=f_j$ where $f_j$ comes from phrase {\scriptsize $(1)$}" if the $j$th rank of the premise is conjunctive, and the phrase {\scriptsize $(4)$} ``there is a $j$-metaselection $g_j$" if the $j$th rank of the premise is disjunctive.

Then, by the definition of truth (of the premise), we have that $\mathcal{O}_1g_1\ldots\mathcal{O}_mg_m$ such that the premise $\Phi\bigl\{\Psi\{\mathcal{A}\}\odot_{k^{i}}\mathcal{C}\bigr\}$ is metatrue w.r.t. $(^{*},(g_1,\ldots,g_m))$. By lemma \ref{yan2}, $\Phi\bigl\{\Psi\{\mathcal{A}\}\odot_{k^{i}}\mathcal{C}\bigr\}$ is metatrue w.r.t. $(^{*},(g_1,\ldots,g_m))$ iff
$g_i(k)=\mbox{\em left}$ (resp. $g_i(k)=\mbox{\em right}$) and
$\Phi\bigl\{\Psi\{\mathcal{A}\}\bigr\}$ (resp. $\Phi\{\mathcal{C}\}$) is metatrue w.r.t. $(^{\ast},(g_1,\ldots,g_m))$. This, in turn, is the case iff $f_i(k)=\mbox{\em left}$ (resp. $f_i(k)=\mbox{\em right}$) and
$\Phi\bigl\{\Psi\{\mathcal{A}\}\bigr\}$ (resp. $\Phi\{\mathcal{C}\}$) is metatrue w.r.t. $(^{\ast},(f_1,\ldots,f_m))$. But the ranks $m+1,\ldots,n$ only occurs in the subcirquent $\mathcal{B}$ of the conclusion. So, $\Phi\bigl\{\Psi\{\mathcal{A}\}\bigr\}$ (resp. $\Phi\{\mathcal{C}\}$) is metatrue w.r.t. $(^{\ast},(f_1,\ldots,f_m))$ iff $\Phi\bigl\{\Psi\{\mathcal{A}\}\bigr\}$ (resp. $\Phi\{\mathcal{C}\}$) is metatrue w.r.t. $(^{\ast},(f_1,\ldots,f_m,f_{m+1},\ldots,f_n))$ for any sequence of metaselections $f_{m+1},\ldots,f_n$. Finally, by lemma \ref{yan2} (applied twice), we get that the premise $\Phi\bigl\{\Psi\{\mathcal{A}\}\odot_{k^{i}}\mathcal{C}\bigr\}$ is metatrue w.r.t. $(^{*},(g_1,\ldots,g_m))$  iff the conclusion $\Phi\bigl\{\Psi\{\mathcal{A}\odot_{k^{i}}\mathcal{B}\}\odot_{k^{i}}\mathcal{C}\bigr\}$ is metatrue w.r.t. $(^{\ast},(f_1,\ldots,f_n))$.

Let $\mathcal{Q}_rf_r$ ($1\leq r\leq n$) be the abbreviation satisfying the conditions that: (a) when $1\leq r\leq m$, $\mathcal{Q}_rf_r=\mathcal{P}_rf_r$; (b) when $m+1\leq r\leq n$, $\mathcal{Q}_rf_r$ abbreviates the phrase ``for every $r$-metaselection $f_r$" if the $r$th rank of the conclusion is conjunctive, and ``Let $f_r$ be any $r$-metaselection" if the $r$th rank of the conclusion is disjunctive. Then we have $\mathcal{Q}_1f_1\ldots\mathcal{Q}_nf_n$ such that the conclusion $\Phi\bigl\{\Psi\{\mathcal{A}\odot_{k^{i}}\mathcal{B}\}\odot_{k^{i}}\mathcal{C}\bigr\}$ is metatrue w.r.t. $(^{\ast},(f_1,\ldots,f_n))$. Hence, the conclusion is true under $^{*}$.\vspace{2mm}

{\em Rule II}: Again, we only consider Rule II (left),  with Rule II (right) being similar. We want to show that the premise $\Phi\bigl\{(\mathcal{A}\circ_{m^{j}}\mathcal{C}_1)\odot_{k^{i}}(\mathcal{B}\circ_{n^{j}}\mathcal{C}_2)\bigr\}$ is true under $^{*}$ iff so is the conclusion $\Phi\bigl\{(\mathcal{A}\odot_{k^{i}}\mathcal{B})\circ_{l^{j}}\mathcal{C}\bigr\}$, where conditions (i)--(iv) are satisfied as described in the preceding introduction of this rule.

Obviously, the premise and the conclusion have the same (number of) ranks. Suppose that both of them have $n$ ranks $1,\ldots,n$.
When $i<j$, there are four situations of the values of $\odot$ and $\circ$, i.e. $\odot$ and $\circ$ are either $\wedge$ and $\vee$, or $\vee$ and $\wedge$, or $\wedge$ and $\wedge$, or $\vee$ and $\vee$, respectively. When $i=j$, $\odot$ and $\circ$ should be the same type, in which case the proposition can be similarly proven as the last two situations when $i<j$. Here we only consider one situation that $\odot$ and $\circ$ are $\wedge$ and $\vee$, respectively, when $i<j$, with the other situations being similar. Two (sub)cases need to be further considered.

{\em Case (a)}:  cluster $l$ is a singleton in the conclusion (which means clusters $m,n$ are singletons in the premise). Let $\{l_1,\ldots,l_n\}$ be the collection of all singleton clusters in the subcirquent $\mathcal{C}$ of the conclusion.
And let $\{l'_1,\ldots,l'_n\}$ (resp. $\{l''_1,\ldots,l''_n\}$) be the collection of all singleton clusters in the subcirquent $\mathcal{C}_1$ (resp. $\mathcal{C}_2$) in the premise satisfying the condition that, for any $h\in\{1,\ldots,n\}$, $l'_h$ (resp. $l''_h$) occurs in $\mathcal{C}_1$ (resp. $\mathcal{C}_2$) at the same place as $l_h$ occurs in the $\mathcal{C}$ part of the conclusion.

Suppose that the premise $\Phi\bigl\{(\mathcal{A}\vee_{m^{j}}\mathcal{C}_1)\wedge_{k^{i}}(\mathcal{B}\vee_{n^{j}}\mathcal{C}_2)\bigr\}$ is true under $^{*}$.
For any $r\in\{1,\ldots,n\}$, let $\mathcal{Q}_rf_r$ abbreviate the phrase {\scriptsize $(1)$} ``for every $r$-metaselection $f_r$" if the $r$th rank of the conclusion is conjunctive, and the phrase {\scriptsize $(2)$} ``let $f_r$ be a $r$-metaselection satisfying the condition that $f_r(l)=g_r(m)$ and $f_r(l_h)=g_r(l'_h)$ for any $h\in\{1,\ldots,n\}$ when $g_i(k)=\mbox {\em left}$, $f_r(l)=g_r(n)$ and $f_r(l_h)=g_r(l''_h)$ for any $h\in\{1,\ldots,n\}$ when $g_i(k)=\mbox {\em right}$, and $f_r$ agrees with $g_r$ on all other clusters, where $g_r$ is the $r$-metaselection in phrase {\scriptsize $(5)$}" if the $r$th rank of the conclusion is disjunctive and $r\geq i$ . And let $\mathcal{Q}_rf_r$ abbreviate the phrase {\scriptsize $(3)$} ``let $f_r$ be the $r$-metaselection $g_r$, where $g_r$ is the one in phrase {\scriptsize $(5)$}" if the $r$th rank of the conclusion is disjunctive and $r<i$. Let $\mathcal{P}_rg_r$ abbreviate the phrase {\scriptsize $(4)$} ``for the $r$-metaselection $g_r$ satisfying the condition that $g_r(l'_h)=g_r(l''_h)=f_r(l_h)$ for any $h\in\{1,\ldots,n\}$ and $g_r$ agrees with $f_r$ on all other clusters, where $f_r$ comes from phrase {\scriptsize $(1)$}" if the $r$th rank of the premise is conjunctive, and the phrase {\scriptsize $(5)$} ``there is a $r$-metaselection $g_r$" if the $r$th rank of the premise is disjunctive.

Then, by the definition of truth (of the premise), we have $\mathcal{P}_1g_1\ldots\mathcal{P}_ng_n$ such that the premise $\Phi\bigl\{(\mathcal{A}\vee_{m^{j}}\mathcal{C}_1)\wedge_{k^{i}}(\mathcal{B}\vee_{n^{j}}\mathcal{C}_2)\bigr\}$ is  metatrue w.r.t. $(^{\ast},(g_1,\ldots,g_n))$. But by lemma \ref{yan2}, the premise $\Phi\bigl\{(\mathcal{A}\vee_{m^{j}}\mathcal{C}_1)\wedge_{k^{i}}(\mathcal{B}\vee_{n^{j}}\mathcal{C}_2)\bigr\}$ is  metatrue w.r.t. $(^{\ast},(g_1,\ldots,g_n))$ iff  $g_i(k)=\mbox{\em left}$ (resp. $g_i(k)=\mbox{\em right}$)  and
$\Phi\{\mathcal{A}\vee_{m^{j}}\mathcal{C}_1\}$ (resp. $\Phi\{\mathcal{B}\vee_{n^{j}}\mathcal{C}_2\}$) is metatrue w.r.t. $(^{\ast},(g_1,\ldots,g_n))$. This, in turn, is the case iff  $f_i(k)=\mbox{\em left}$ (resp. $f_i(k)=\mbox{\em right}$)  and
$\Phi\{\mathcal{A}\vee_{l^{j}}\mathcal{C}\}$ (resp. $\Phi\{\mathcal{B}\vee_{l^{j}}\mathcal{C}\}$) is metatrue w.r.t. $(^{\ast},(f_1,\ldots,f_n))$.\footnote{Note that all the ranks $t$
that occurs in the subcirquent $\mathcal{C}$ (resp. $\mathcal{C}_1$ or $\mathcal{C}_2$) of the conclusion (resp. the premise) satisfy the condition that $t\geq i$.} The above, in turn, is the case iff the conclusion
$\Phi\bigl\{(\mathcal{A}\wedge_{k^{i}}\mathcal{B})\vee_{l^{j}}\mathcal{C}\bigr\}$ is metatrue w.r.t. $(^{\ast},(f_1,\ldots,f_n))$. Finally, we have $\mathcal{Q}_1f_1\ldots\mathcal{Q}_nf_n$ such that the conclusion is metatrue w.r.t. $(^{\ast},(f_1,\ldots,f_n))$. Hence, the conclusion is true under $^{*}$.

Now suppose that the conclusion $\Phi\bigl\{(\mathcal{A}\wedge_{k^{i}}\mathcal{B})\vee_{l^{j}}\mathcal{C}\bigr\}$ is true under $^{*}$. For any $r\in\{1,\ldots,n\}$, let $\mathcal{Q}_rg_r$ abbreviate the phrase {\scriptsize $(1)$} ``for every $r$-metaselection $g_r$" if the $r$th rank of the premise is conjunctive, and the phrase {\scriptsize $(2)$} ``let $g_r$ be a $r$-metaselection satisfying the condition that $g_r(m)=g_r(n)=f_r(l)$, $g_r(l'_h)=g_r(l''_h)=f_r(l_h)$ for any $h\in\{1,\ldots,n\}$, and $g_r$ agrees with $f_r$ on all other clusters, where $f_r$ is the $r$-metaselection in phrase {\scriptsize $(5)$}" if the $r$th rank of the premise is disjunctive. Let $\mathcal{P}_rf_r$ abbreviate the phrase {\scriptsize $(3)$} ``for the $r$-metaselection $f_r=g_r$, where $g_r$ comes from phrase {\scriptsize $(1)$}" if the $r$th rank of the conclusion is conjunctive and $r<i$, and the phrase {\scriptsize $(4)$} ``for the $r$-metaselection $f_r$ satisfying the condition that $f_r(l_h)=g_r(l'_h)$ for any $h\in\{1,\ldots,n\}$ when $g_i(k)=\mbox {\em left}$, $f_r(l_h)=g_r(l''_h)$ for any $h\in\{1,\ldots,n\}$  when $g_i(k)=\mbox {\em right}$, and $f_r$ agrees with $g_r$ on all other clusters, where $g_r$ comes from phrase {\scriptsize $(1)$}" if the $r$th rank of the conclusion is conjunctive and $r\geq i$. And let $\mathcal{P}_rf_r$ abbreviate the phrase {\scriptsize $(5)$} ``there is a $r$-metaselection $f_r$" if the $r$th rank of the conclusion is disjunctive.

Then, by the definition of truth (of the conclusion), we have $\mathcal{P}_1f_1\ldots\mathcal{P}_nf_n$ such that the conclusion $\Phi\bigl\{(\mathcal{A}\wedge_{k^{i}}\mathcal{B})\vee_{l^{j}}\mathcal{C}\bigr\}$ is metatrue w.r.t. $(^{\ast},(f_1,\ldots,f_n))$. By lemma \ref{yan2}, the conclusion $\Phi\bigl\{(\mathcal{A}\wedge_{k^{i}}\mathcal{B})\vee_{l^{j}}\mathcal{C}\bigr\}$ is metatrue w.r.t. $(^{\ast},(f_1,\ldots,f_n))$ iff $f_i(k)=\mbox{\em left}$ (resp. $f_i(k)=\mbox{\em right}$)  and
$\Phi\{\mathcal{A}\vee_{l^{j}}\mathcal{C}\}$ (resp. $\Phi\{\mathcal{B}\vee_{l^{j}}\mathcal{C}\}$) is metatrue w.r.t. $(^{\ast},(f_1,\ldots,f_n))$. This, in turn, is the case iff $g_i(k)=\mbox{\em left}$ (resp. $g_i(k)=\mbox{\em right}$)  and
$\Phi\{\mathcal{A}\vee_{m^{j}}\mathcal{C}_1\}$ (resp. $\Phi\{\mathcal{B}\vee_{n^{j}}\mathcal{C}_2\}$) is metatrue w.r.t. $(^{\ast},(g_1,\ldots,g_n))$.\footnote{The same to the preceding note.} The above, in turn, is the case iff the premise $\Phi\bigl\{(\mathcal{A}\vee_{m^{j}}\mathcal{C}_1)\wedge_{k^{i}}(\mathcal{B}\vee_{n^{j}}\mathcal{C}_2)\bigr\}$ is metatrue w.r.t. $(^{\ast},(g_1,\ldots,g_n))$. Finally, we have $\mathcal{Q}_1g_1\ldots\mathcal{Q}_ng_n$ such that the premise is metatrue w.r.t. $(^{\ast},(g_1,\ldots,g_n))$. Hence, the premise is true under $^{*}$.

{\em Case (b)}: cluster $l$ is a non-singleton in the conclusion (which means $m=n=l$). By employing an essentially similar (but simpler) argument as we did in {\em Case (a)}, this case can be easily proven, so we leave details to the reader.

\vspace{2mm}
{\em Rule III}: We want to show that the premise $\Phi\bigl\{(\mathcal{A}\circ_{m^{j}}\mathcal{C})\odot_{k^{i}}(\mathcal{B}\circ_{n^{j}}\mathcal{D})\bigr\}$ is true under $^{*}$ iff so is the conclusion  $\Phi\bigl\{(\mathcal{A}\odot_{k^{i}}\mathcal{B})\circ_{l^{j}}(\mathcal{C}\odot_{k^{i}}\mathcal{D})\bigr\}$, where conditions (i),(ii) are satisfied as described in the preceding introduction of this rule.
Suppose that both the premise and the conclusion have (the same) $n$ ranks $1,\ldots,n$. Again, we only consider one situation that both $\odot$ and $\circ$ are $\vee$ and $i<j$, with the other situations being similar. Two (sub)cases are further considered below.

{\em Case (a)}: cluster $l$ is a singleton in the conclusion (which means clusters $m,n$ are singletons in the premise).

Suppose that the premise $\Phi\bigl\{(\mathcal{A}\vee_{m^{j}}\mathcal{C})\vee_{k^{i}}(\mathcal{B}\vee_{n^{j}}\mathcal{D})\bigr\}$ is true under $^{*}$.
For any $r\in\{1,\ldots,n\}$, let $\mathcal{Q}_rf_r$ abbreviate the phrase {\scriptsize (1)} ``for every $r$-metaselection $f_r$" if the $r$th rank of the conclusion is conjunctive, and the phrase {\scriptsize (2)} ``let $f_r$ be the $r$-metaselection $g_r$, where $g_r$ is the one in phrase {\scriptsize (5)}" if the $r$th rank of the conclusion is disjunctive and $r\leq i$, and the phrase {\scriptsize (3)} ``let $f_r$ be a $r$-metaselection satisfying the condition that $f_r(l)=g_r(m)$ when $g_i(k)=\mbox{\em left}$, $f_r(l)=g_r(n)$ when $g_i(k)=\mbox{\em right}$, and $f_r$ agrees with $g_r$ on all other clusters, where $g_r$ is the $r$-metaselection in phrase {\scriptsize (5)}" if the $r$th rank of the conclusion is disjunctive and $r>i$.
Let $\mathcal{P}_rg_r$ abbreviate the phrase {\scriptsize (4)} ``for the $r$-metaselection $g_r=f_r$, where $f_r$ comes from phrase {\scriptsize (1)}" if the $r$th rank of the premise is conjunctive, and the phrase {\scriptsize (5)} ``there is a $r$-metaselection $g_r$" if the $r$th rank of the premise is disjunctive.

Then, by the definition of truth (of the premise), we have that $\mathcal{P}_1g_1\ldots\mathcal{P}_ng_n$ such that the premise
$\Phi\bigl\{(\mathcal{A}\vee_{m^{j}}\mathcal{C})\vee_{k^{i}}(\mathcal{B}\vee_{n^{j}}\mathcal{D})\bigr\}$ is metatrue w.r.t. $(^{\ast},(g_1,\ldots,g_n))$. By lemma \ref{yan2}, the premise $\Phi\bigl\{(\mathcal{A}\vee_{m^{j}}\mathcal{C})\vee_{k^{i}}(\mathcal{B}\vee_{n^{j}}\mathcal{D})\bigr\}$ is metatrue w.r.t. $(^{\ast},(g_1,\ldots,g_n))$ iff $g_i(k)=\mbox{\em left}$ (resp. $g_i(k)=\mbox{\em right}$)  and
$\Phi\bigl\{\mathcal{A}\vee_{m^{j}}\mathcal{C}\bigr\}$ (resp.  $\Phi\bigl\{\mathcal{B}\vee_{n^{j}}\mathcal{D}\bigr\}$) is metatrue w.r.t. $(^{\ast},(g_1,\ldots,g_n))$. This, in turn, is the case iff  $f_i(k)=\mbox{\em left}$ (resp. $f_i(k)=\mbox{\em right}$)  and
$\Phi\bigl\{\mathcal{A}\vee_{l^{j}}\mathcal{C}\bigr\}$ (resp.  $\Phi\bigl\{\mathcal{B}\vee_{l^{j}}\mathcal{D}\bigr\}$) is metatrue w.r.t. $(^{\ast},(f_1,\ldots,f_n))$. Further, by lemma \ref{yan2} (applied twice), the above is the case iff the conclusion $\Phi\bigl\{(\mathcal{A}\vee_{k^{i}}\mathcal{B})\vee_{l^{j}}(\mathcal{C}\vee_{k^{i}}\mathcal{D})\bigr\}$ is metatrue w.r.t. $(^{\ast},(f_1,\ldots,f_n))$. Finally, we have that $\mathcal{Q}_1f_1\ldots\mathcal{Q}_nf_n$ such that the conclusion is metatrue w.r.t. $(^{\ast},(f_1,\ldots,f_n))$. Hence the conclusion is true under $^{*}$.

Now suppose that the conclusion  $\Phi\bigl\{(\mathcal{A}\vee_{k^{i}}\mathcal{B})\vee_{l^{j}}(\mathcal{C}\vee_{k^{i}}\mathcal{D})\bigr\}$ is true under $^{*}$. For any $r\in\{1,\ldots,n\}$, let $\mathcal{Q}_rg_r$ abbreviate the phrase {\scriptsize (1)} ``for every $r$-metaselection $g_r$" if the $r$th rank of the premise is conjunctive, and the phrase {\scriptsize (2)} ``let $g_r$ be the $r$-metaselection $f_r$, where $f_r$ is the one in phrase {\scriptsize (5)}" if the $r$th rank of the premise is disjunctive and $r\leq i$, and the phrase {\scriptsize (3)} ``let $g_r$ be a $r$-metaselection satisfying the condition that $g_r(m)=g_r(n)=f_r(l)$ and $g_r$ agrees with $f_r$ on all other clusters, where $f_r$ is the $r$-metaselection in phrase {\scriptsize (5)}" if the $r$th rank of the premise is disjunctive and $r>i$.
Let $\mathcal{P}_rf_r$ abbreviate the phrase {\scriptsize (4)} ``for the $r$-metaselection $f_r=g_r$, where $g_r$ comes from phrase {\scriptsize (1)}" if the $r$th rank of the conclusion is conjunctive, and the phrase {\scriptsize (5)} ``there is a $r$-metaselection $f_r$" if the $r$th rank of the conclusion is disjunctive.

Then, by the definition of truth (of the conclusion), we have $\mathcal{P}_1f_1\ldots\mathcal{P}_nf_n$ such that the conclusion $\Phi\bigl\{(\mathcal{A}\vee_{k^{i}}\mathcal{B})\vee_{l^{j}}(\mathcal{C}\vee_{k^{i}}\mathcal{D})\bigr\}$ is metatrue w.r.t. $(^{\ast},(f_1,\ldots,f_n))$. By lemma \ref{yan2} (applied twice), the conclusion $\Phi\bigl\{(\mathcal{A}\vee_{k^{i}}\mathcal{B})\vee_{l^{j}}(\mathcal{C}\vee_{k^{i}}\mathcal{D})\bigr\}$ is metatrue w.r.t. $(^{\ast},(f_1,\ldots,f_n))$ iff $f_i(k)=\mbox{\em left}$ (resp. $f_i(k)=\mbox{\em right}$)  and
$\Phi\bigl\{\mathcal{A}\vee_{l^{j}}\mathcal{C}\bigr\}$ (resp.  $\Phi\bigl\{\mathcal{B}\vee_{l^{j}}\mathcal{D}\bigr\}$) is metatrue w.r.t. $(^{\ast},(f_1,\ldots,f_n))$. This, in turn, is the case iff $g_i(k)=\mbox{\em left}$ (resp. $g_i(k)=\mbox{\em right}$)  and
$\Phi\bigl\{\mathcal{A}\vee_{m^{j}}\mathcal{C}\bigr\}$ (resp.  $\Phi\bigl\{\mathcal{B}\vee_{n^{j}}\mathcal{D}\bigr\}$) is metatrue w.r.t. $(^{\ast},(g_1,\ldots,g_n))$. The above, in turn, is the case iff the premise $\Phi\bigl\{(\mathcal{A}\vee_{m^{j}}\mathcal{C})\vee_{k^{i}}(\mathcal{B}\vee_{n^{j}}\mathcal{D})\bigr\}$ is metatrue w.r.t. $(^{\ast},(g_1,\ldots,g_n))$. Hence we have $\mathcal{Q}_1g_1\ldots\mathcal{Q}_ng_n$ such that the premise
$\Phi\bigl\{(\mathcal{A}\vee_{m^{j}}\mathcal{C})\vee_{k^{i}}(\mathcal{B}\vee_{n^{j}}\mathcal{D})\bigr\}$ is metatrue w.r.t. $(^{\ast},(g_1,\ldots,g_n))$, and hence the premise is true under  $^{*}$.

{\em Case (b)}: cluster $l$ is a non-singleton in the conclusion. This case can be proven in a similar (but simpler) way as we did in {\em Case (a)}, whose verification is left to the reader.

\vspace{2mm}
{\em Rule IV}: We want to show that, when $i<j$, the premise $\Phi\bigl\{\mathcal{A}[l^{j}]\odot_{k^{i}}\mathcal{B}[l^{j}/r^{j}]\bigr\}$
is true under $^{*}$ iff so is the conclusion $\Phi\bigl\{\mathcal{A}[l^{j}]\odot_{k^{i}}\mathcal{B}[l^{j}]\bigr\}$, where conditions (i),(ii) are satisfied as described in the preceding introduction of this rule.
Here we only consider the case when rank $i$ is conjunctive and rank $j$ is disjunctive, with the other cases being similar.
Obviously, the premise and the conclusion has the same (number of) ranks. Suppose that both of them have $n$ ranks $1,\ldots,n$.

Suppose that the premise $\Phi\bigl\{\mathcal{A}[l^{j}]\wedge_{k^{i}}\mathcal{B}[l^{j}/r^{j}]\bigr\}$ is true under $^{*}$. For any $t\in\{1,\ldots,n\}$, let $\mathcal{Q}_tf_t$ abbreviate the phrase {\scriptsize $(1)$} ``for every $t$-metaselection $f_t$" if the $t$th rank of the conclusion is conjunctive, and the phrase {\scriptsize $(2)$} ``let $f_t$ be the $t$-metaselection $g_t$, where $g_t$ is the one in phrase {\scriptsize $(5)$}" if the $t$th rank of the conclusion is disjunctive and $t\neq j$, and the phrase {\scriptsize $(3)$} ``let $f_t$ be the $t$-metaselection satisfying the condition that $f_t(l)=g_t(l)$ when $f_i(k)=\mbox {\em left}$, $f_t(l)=g_t(r)$ when $f_i(k)=\mbox {\em right}$ and $f_t$ agrees with $g_t$ on all other clusters, where $g_t$ is the $t$-metaselection in phrase {\scriptsize $(5)$}" if the $t$th rank of the conclusion is disjunctive and $t=j$. Let $\mathcal{P}_tg_t$ abbreviate the phrase {\scriptsize $(4)$} ``for the $t$-metaselection $g_t=f_t$ where $f_t$ comes from phrase {\scriptsize $(1)$}" if the $t$th rank of the premise is conjunctive, and the phrase {\scriptsize $(5)$} ``there is a $t$-metaselection $g_t$" if the $t$th rank of the premise is disjunctive.

Then, by the definition of truth (of the premise), we have $\mathcal{P}_1g_1\ldots\mathcal{P}_ng_n$ such that the premise $\Phi\bigl\{\mathcal{A}[l^{j}]\wedge_{k^{i}}\mathcal{B}[l^{j}/r^{j}]\bigr\}$ is metatrue w.r.t. $(^{*},(g_1,\ldots,g_n))$. By lemma \ref{yan2}, the premise $\Phi\bigl\{\mathcal{A}[l^{j}]\wedge_{k^{i}}\mathcal{B}[l^{j}/r^{j}]\bigr\}$ is metatrue w.r.t. $(^{*},(g_1,\ldots,g_n))$ iff  $g_i(k)=\mbox{\em left}$ (resp. $g_i(k)=\mbox{\em right}$)  and $\Phi\bigl\{\mathcal{A}[l^{j}]\bigr\}$ (resp. $\Phi\bigl\{\mathcal{B}[l^{j}/r^{j}]\bigr\}$) is metatrue w.r.t. $(^{\ast},(g_1,\ldots,g_n))$. This, in turn, is the case iff  $f_i(k)=\mbox{\em left}$ (resp. $f_i(k)=\mbox{\em right}$)  and $\Phi\bigl\{\mathcal{A}[l^{j}]\bigr\}$ (resp. $\Phi\bigl\{\mathcal{B}[l^{j}]\bigr\}$) is metatrue w.r.t. $(^{\ast},(f_1,\ldots,f_n))$.\footnote{Note that $l$ only occurs in the subcirquents $\mathcal{A}$ and $\mathcal{B}$ in the conclusion, and that $i<j$.} The above, in turn, is the case iff the conclusion $\Phi\bigl\{\mathcal{A}[l^{j}]\wedge_{k^{i}}\mathcal{B}[l^{j}]\bigr\}$ is  metatrue w.r.t. $(^{\ast},(f_1,\ldots,f_n))$. Finally, we have $\mathcal{Q}_1f_1\ldots\mathcal{Q}_nf_n$ such that the conclusion is metatrue w.r.t. $(^{\ast},(f_1,\ldots,f_n))$. Hence, the conclusion is true under $^{*}$.

Now suppose that the conclusion $\Phi\bigl\{\mathcal{A}[l^{j}]\wedge_{k^{i}}\mathcal{B}[l^{j}]\bigr\}$ is true under $^{*}$. For any $t\in\{1,\ldots,n\}$, let $\mathcal{Q}_tg_t$ abbreviate the phrase {\scriptsize $(1)$} ``for every $t$-metaselection $g_t$" if the $t$th rank of the premise is conjunctive, and the phrase {\scriptsize $(2)$} ``let $g_t$ be the $t$-metaselection $f_t$, where $f_t$ is the one in phrase {\scriptsize $(5)$}" if the $t$th rank of the premise is disjunctive and $t\neq j$, and the phrase {\scriptsize $(3)$} ``let $g_t$ be the $t$-metaselection satisfying the condition that $g_t(r)=f_t(l)$ and $g_t$ agrees with $f_t$ on all other clusters, where $f_t$ is the $t$-metaselection in phrase {\scriptsize $(5)$}" if the $t$th rank of the premise is disjunctive and $t=j$. Let $\mathcal{P}_tf_t$ abbreviate the phrase {\scriptsize $(4)$} ``for the $t$-metaselection $f_t=g_t$ where $g_t$ comes from phrase {\scriptsize $(1)$}" if the $t$th rank of the conclusion is conjunctive, and the phrase {\scriptsize $(5)$} ``there is a $t$-metaselection $f_t$" if the $t$th rank of the conclusion is disjunctive.

Then, by the definition of truth (of the conclusion), we have $\mathcal{P}_1f_1\ldots\mathcal{P}_nf_n$ such that the conclusion $\Phi\bigl\{\mathcal{A}[l^{j}]\wedge_{k^{i}}\mathcal{B}[l^{j}]\bigr\}$ is  metatrue w.r.t. $(^{\ast},(f_1,\ldots,f_n))$. But by lemma \ref{yan2}, the conclusion $\Phi\bigl\{\mathcal{A}[l^{j}]\wedge_{k^{i}}\mathcal{B}[l^{j}]\bigr\}$ is  metatrue w.r.t. $(^{\ast},(f_1,\ldots,f_n))$ iff $f_i(k)=\mbox{\em left}$ (resp. $f_i(k)=\mbox{\em right}$)  and $\Phi\bigl\{\mathcal{A}[l^{j}]\bigr\}$ (resp. $\Phi\bigl\{\mathcal{B}[l^{j}]\bigr\}$) is metatrue w.r.t. $(^{\ast},(f_1,\ldots,f_n))$. This, in turn, is the case iff $g_i(k)=\mbox{\em left}$ (resp. $g_i(k)=\mbox{\em right}$)  and $\Phi\bigl\{\mathcal{A}[l^{j}]\bigr\}$ (resp. $\Phi\bigl\{\mathcal{B}[l^{j}/r^{j}]\bigr\}$) is metatrue w.r.t. $(^{\ast},(g_1,\ldots,g_n))$.\footnote{The same to the preceding note.} The above, in turn, is the case iff the premise $\Phi\bigl\{\mathcal{A}[l^{j}]\wedge_{k^{i}}\mathcal{B}[l^{j}/r^{j}]\bigr\}$ is metatrue w.r.t. $(^{*},(g_1,\ldots,g_n))$. Hence we have $\mathcal{Q}_1g_1\ldots\mathcal{Q}_ng_n$ such that the premise is  metatrue w.r.t. $(^{*},(g_1,\ldots,g_n))$, and hence the premise is true under $^{*}$.\end{proof}

\subsection{The soundness and completeness of $RIF_p$}
\begin{theorem}
A cirquent is valid if and only if it is provable in $RIF_p$.
\end{theorem}
\begin{proof}
The soundness part is immediate, because the axioms are obviously valid and, by Lemma \ref{niu}, all rules preserve truth and hence validity. For the completeness part, consider an arbitrary cirquent $\mathcal{C}$ and assume it is valid. We want to show that $\mathcal{C}$ is provable in $RIF_p$.\vspace{2mm}

In the context of a given cirquent $\mathcal{D}$, we define the {\bf level} of an oconnective $a$, denoted by $\mathcal{L}^{\mathcal{D}}(a)$, to be the total number of oconnectives $b$ such that $a$ is in the scope of $b$. An oconnective $b$ is a {\bf child} of an oconnective $a$ and $a$ is the {\bf parent} of $b$ when $b$ is in the scope of $a$ and $\mathcal{L}^{\mathcal{D}}(b) = \mathcal{L}^{\mathcal{D}}(a)+1$. The relations ``descendant" and ``ancestor" are the transitive closures of the relations ``child" and ``parent", respectively. The {\bf distance} between an oconnective $a$ and one of its descendants $b$ is defined to be the positive integer $k$ such that $k = \mathcal{L}^{\mathcal{D}}(b)-\mathcal{L}^{\mathcal{D}}(a)$; when the distance between $a$ and $b$ is less than the distance between $a$ and another descendant $c$ of $a$, we say that $b$ is {\bf nearer} to $a$ (or vice versa) than $c$ is. Next, for any two oconnectives $a$ and $b$, we denote their nearest common ancestor oconnective by  $\underline{ab}$.\vspace{2mm}

If our cirquent $\mathcal{C}$ is classical (i.e. every cluster of it is a singleton),  then  the validity of $\mathcal{C}$ can be seen to mean nothing but its validity in the sense of classical logic. So, in this case,  $\mathcal{C}$ is an axiom of $RIF_p$ and hence is provable.

Now, for the rest of this proof, assume that $\mathcal{C}$ is not classical and that it has $n$ ranks $1,\ldots,n$. We construct, bottom-up, a proof of $\mathcal{C}$ as follows. We will use $\mathcal{D}$ to denote the current (topmost in the so far constructed proof) cirquent and, for convenience of descriptions, when an oconnective $a$ is in (whatever) cluster $k$ which is of rank $i$, we also say that $a$ is in rank $i$.\vspace{2mm}

{\bf Step 1}: Let $i$ be a variable during this step to denote the number of iterations of the outmost loop of Step 1.
For $i=1$ to $n$, do the following: repeat the step below while there is an oconnective $c$ in $\mathcal{D}$ such that $c$ is in rank $i$ and that an ancestor of $c$ is in rank $j$ while $j>i$:
\begin{itemize}
\item Let $S$ be a collection of all oconnectives $b$ in $\mathcal{D}$ satisfying the condition that $b$ is in rank $i$ and that an ancestor of $b$ is in rank $j$ with $j>i$. Pick an oconnective $a$  in $S$ such that $\mathcal{L}^{\mathcal{D}}(a)\leq\mathcal{L}^{\mathcal{D}}(b)$ for any $b$ in $S$.
Repeat the following two steps while the parent of $a$ is in rank $j$ with $j>i$:
\begin{itemize}
\item Apply (bottom-up) Rule II to $\mathcal{D}$, with $a$ being the key oconnective of this application (in the conclusion);
\item Rename the key oconnective of this application (in the premise) into $a$.
\end{itemize}
\end{itemize}

With some thoughts, one can see that Rule II can always be applied to the current cirquent while ``the parent of $a$ is in rank $j$ with $j>i$". Note that every time Rule II is (bottom-up) applied, the level of the picked oconnective $a$ in the current cirquent is decreased by 1. Below we show that, for any fixed $i$, the $i$th loop of Step 1 terminates in finite steps.

For the current cirquent $\mathcal{D}$ at any given stage of the $i$th loop of Step 1 --- hencefore we shall use $\mathcal D$ as (also) a name of that stage ---  we define the {\bf $i$-distribution} of $\mathcal {D}$ to be an infinite sequence $(x_0,x_1,x_2,\ldots)$ where, for any $m\in\{0,1,2,\ldots\}$, $x_m$ is the total number of oconnectives $c$ such that $c$ is in rank $i$ and that $\mathcal{L}^{\mathcal{D}}(c)=m$.

Further, we define the relation ``$\leq$" on the set of all such sequences as follows. For any two sequences $(x_0,x_1,x_2,\ldots)$ and $(y_0,y_1,y_2,\ldots)$, $(x_0,x_1,x_2,\ldots)\leq(y_0,y_1,y_2,\ldots)$ if and only if one of the following conditions holds: for any $m\in\{0,1,2,\ldots\}$, $x_m=y_m$; $x_0> y_0$; $x_0=y_0$ and $x_1> y_1$; $x_0=y_0$, $x_1=y_1$ and $x_2> y_2$; $x_0=y_0$, $x_1=y_1$, $x_2=y_2$ and $x_3> y_3$; $\ldots$. It is easy to see that ``$\leq$" well-orders the set of all $i$-distributions, with each sequence ($i$-distribution) denoting an ordinal $<\omega^{\omega}$.

Then, we can see that the $i$-distribution of the current cirquent keeps strictly decreasing during the $i$th loop of Step 1, meaning that the latter terminates at some point. Hence, the above Step 1 terminates and we get a cirquent $\mathcal{C}_1$ where there is no oconnecitves $c$ such that $c$ is in rank $i$ and an ancestor of $c$ is in rank $j$ with $j>i$. We call this property that $\mathcal{C}_1$ have the {\bf property 1} for later reference. Since Rule II preserves truth in the bottom-up direction (by Lemma \ref{niu}) and $\mathcal{C}$ is valid, $\mathcal{C}_1$ is valid. Then, our construction of a proof of $\mathcal{C}$ continues upward from $\mathcal{C}_{1}$ as follows.\vspace{2mm}

{\bf Step 2}: Repeat applying (bottom-up) Rule I to $\mathcal{D}$ until no longer possible.\vspace{2mm}

Every time Rule I is applied, the current cirquent loses one pair of oconnectives $a,b$ such that $a,b$ are in the same cluster and $b$ is a descendant of $a$ (or vice versa). So, sooner or later, we get a cirquent $\mathcal{C}_{2}$ where no descendant-ancestor pair of oconnectives shares the same cluster. We call this property that $\mathcal{C}_2$ have the {\bf property 2} for later reference. Obviously, $\mathcal{C}_2$ also have the property 1. Since Rule I preserves truth in the bottom-up direction (Lemma \ref{niu}), $\mathcal{C}_{2}$ is valid. Next, our construction continues as follows.\vspace{2mm}

{\bf Step 3}: Let $i$ be a variable during this step to denote the number of iterations of the outmost loop of Step 3.
For $i=1$ to $n$, do the following:
\begin{itemize}
\item{\bf 3.1} Repeat the step below while there are non-singleton clusters of rank $i$ in $\mathcal{D}$.
\begin{itemize}
\item Pick any cluster $k$ from the collection of all non-singleton clusters of rank $i$ in $\mathcal{D}$. Repeat the following four steps until cluster $k$ becomes a singleton cluster in $\mathcal{D}$.
\begin{itemize}
\item{\bf 3.1.1} Pick any pair $a,b$ of oconnectives such that the following two conditions are satisfied: $a$, $b$ are both in cluster $k$; $\mathcal{L}^{\mathcal{D}}(\underline{ab})\geq\mathcal{L}^{\mathcal{D}}(\underline{cd})$ for any pair of oconnectives $c, d$ in cluster $k$. Set $m=2$ and $l=\mathcal{L}^{\mathcal{D}}(\underline{ab})$.\footnote{Here $m$ is a variable that records the number of elements of the collection of key oconnectives in the current cirquent. It is introduced into the process mainly for later difinitions and proofs.}

\item{\bf 3.1.2} Repeatedly perform the following two actions until $\mathcal{L}^{\mathcal{D}}(a)=l+1$: (i) Apply (bottom-up) Rule II to $\mathcal{D}$, with $a$ being the key oconnective of this application (in the conclusion); (ii) Rename the key oconnective of this application (in the premise) into $a$.

\item{\bf 3.1.3} Repeatedly perform the following two actions until $\mathcal{L}^{\mathcal{D}}(b)=l+1$: (i) Apply (bottom-up) Rule II to $\mathcal{D}$, with $b$ being the key oconnective of this application (in the conclusion); (ii) Rename the key oconnective of this application (in the premise) into $b$.

\item{\bf 3.1.4} Apply (bottom-up) Rule III to $\mathcal{D}$, with $a,b$ being the key oconnectives of this application (in the conclusion). Set $m=1$.
\end{itemize}
\end{itemize}
\item{\bf 3.2} Label all the oconnectives in rank $i$ of $\mathcal{D}$ with ``unused" and repeat the following step until there is no oconnectives in $\mathcal{D}$ labeled with ``unused":
\begin{itemize}\item Pick an unused oconnective $c$ in rank $i$ of $\mathcal{D}$ such that $\mathcal{L}^{\mathcal{D}}(c)\leq\mathcal{L}^{\mathcal{D}}(d)$ for any unused oconnective $d$ in rank $i$ of $\mathcal{D}$. Repeat the following until no longer possible and label $c$ with ``used":
 \begin{itemize}
 \item Apply (bottom-up) Rule IV to $\mathcal{D}$, with $c$ being the key oconnective of this application (in the conclusion);
 \item Rename the key oconnective of this application (in the premise) into $c$.
 \end{itemize}
\end{itemize}
\end{itemize}

During the above Step 3, property 1 always holds for the current cirquent $\mathcal{D}$. Below are the reasons. At the beginning of Step 3, property 1 holds for $\mathcal{D}$. Then, with some thought, one can see that in step 3.1.1 the oconnective $\underline{ab}$ for the picked $a$, $b$ is also in rank $i$ (otherwise, say, if $\underline{ab}$ is in rank $j$ with $j<i$, then Rule IV can be applied (bottom-up) to the current cirquent with $\underline{ab}$ being the key oconnective in the conclusion, which contradicts what the preceding (sub)steps 3.2 have done). So, each oconnective $c$, which is both an ancestor of $a$ (resp. $b$) and a descendant of $\underline{ab}$, is also in rank $i$. Thus, applying Rule II to $\mathcal{D}$ as described in step 3.1.2 (resp. 3.1.3) will not destroy its property 1. For the same reason, applying Rule III to $\mathcal{D}$ in step 3.1.4 will neither destroy its property 1. Finally, applying Rule IV to $\mathcal{D}$ in step 3.2 obviously will not destroy its property 1, either. This property ensures that, in the $i$th iteration of the outmost loop of Step 3 when Rule II is applied to the current cirquent in step 3.1.2 (resp. 3.1.3), $a$ (resp. $b$) and its parent $e$ are in the same rank $i$ and all the descendent of $e$ are in ranks $j$ with $j\geq i$. And hence Rule II can always be applied in step 3.1.2 (resp. 3.1.3). Similarly, this property also ensues that $a$ and $b$ are in the same rank as $\underline{ab}$ is at the beginning of step 3.1.4, and hence Rule III can always be applied in this step.

During the above Step 3, property 2 also holds for the current cirquent $\mathcal{D}$. To see why, note that there are no descendant-ancestor pairs within any given cluster in $\mathcal{D}$ at the beginning of Step 3. If so, picking any pair $a$, $b$ of oconnectives in step 3.1.1 will not give rise to the situation. Next, applying Rule II to $\mathcal{D}$ in step 3.1.2 (resp. 3.1.3) means that $\mathcal{L}^{\mathcal{D}}(a)$ (resp. $\mathcal{L}^{\mathcal{D}}(b)$) is greater than $\mathcal{L}^{\mathcal{D}}(\underline{ab})+1$; fourthly, during steps 3.1.1, 3.1.2 and 3.1.3, $l$'s being maximal ensures that, when applying Rule II as its introductory figure shows to the current cirquent $\mathcal{D}$, the subcirquent $\mathcal{C}$ does not contain any oconnective $e$ that is also in cluster $k$ (otherwise, $\mathcal{L}^{\mathcal{D}}(\underline{ae})>\mathcal{L}^{\mathcal{D}}(\underline{ab})$ (resp. $\mathcal{L}^{\mathcal{D}}(\underline{be})>\mathcal{L}^{\mathcal{D}}(\underline{ab})$), which contradicts the conditions satisfied by $a,b$). Finally, based on the above conditions, applying Rule III in step 3.1.4 and applying Rule IV in step 3.2 will not create any descendant-ancestor pairs within any given cluster, either. This property ensures that, when the current cirquent has non-singleton clusters $k$ of any rank $i$ during Step 3, the relation between any two different oconnectives in cluster $k$ will not be a descendant-ancestor pair.

Now we verify that Step 3 terminates in a finite number of steps, as shown in the following (i)---(iii).

(i) Firstly, we claim that, for the picked cluster $k$ from the collection of all non-singleton clusters of rank $i$ in the current cirquent $\mathcal{D}$, the four-step procedure (i.e. steps 3.1.1---3.1.4) terminates in a finite number of steps. To see why, we give the following definition.

For the current cirquent $\mathcal{D}$ at any given stage of the four-step procedure, we define the {\bf state} of $\mathcal {D}$ to be the four-tuple $(x,y,z,t)$ as follows, where $m_{\cal D}$, $a_{\cal D}$, $b_{\cal D}$ are the values of the corresponding three variables of the procedure at the beginning of stage $\cal D$:
\begin{itemize}
  \item  $x$ is the number of elements in cluster $k$ of rank $i$ of $\mathcal{D}$;
  \item $y=x-m_{\mathcal D}$;
  \item $z=\mathcal{L}^{\mathcal{D}}(a_{\cal D})+\mathcal{L}^{\mathcal{D}}(b_{\cal D})$ if $m_{\mathcal{D}}=2$, and $z=\mathcal{L}^{\mathcal{D}}(a_{\cal D})-1$ if $m_{\mathcal{D}}=1$;
  \item $t$ is the total number of elements in all other non-singleton clusters of rank $i$ of $\mathcal{D}$ except cluster $k$.
  \end{itemize}

Further, we define the relation ``$\leq$" on the set of all such tuples as follows. For any two tuples $(x_1,y_1,z_1,t_1)$ and $(x_2,y_2,z_2,t_2)$, $(x_1,y_1,z_1,t_1)\leq(x_2,y_2,z_2,t_2)$ if and only if one of the following conditions holds: (i) $x_1< x_2$; (ii) $x_1=x_2$ and $y_1< y_2$; (iii) $x_1=x_2$, $y_1=y_2$ and $z_1< z_2$; (iv) $x_1=x_2$, $y_1=y_2$, $z_1=z_2$ and $t_1< t_2$; (v) $x_1=x_2$, $y_1=y_2$, $z_1=z_2$ and $t_1=t_2$. It is easy to see that ``$\leq$" well-orders the set of all states, with each tuple (state) denoting an ordinal $<\omega^{4}$.

Now we show that every step of the four-step procedure strictly decreases the state of the current cirquent. One can see that the state of the current cirquent depends on cluster $k$ of rank $i$ and the picked pair of $a,b$. At the beginning of this procedure, for the picked cluster $k$ and the picked pair of $a,b$ in step 3.1.1, the state $(x,y,z,t)$ has its original value. In step 3.1.2 (resp. 3.1.3), every time substeps (i),(ii) are performed for $a$ (resp. $b$), $x,y$ do not change, but $z$ decreases by 1; then in step 3.1.4, when Rule III is applied, $x$ is decreased by 1. As long as cluster $k$ of rank $i$ is not a singleton, the process will come into the next iteration of the loop. During each iteration, a new pair of $a,b$ is picked and the value of $m$ is changed from $1$ to $2$ in step 3.1.1, which makes $x$ unchanged but $y$ is decreased by 1; then every iteration of the inner loop in step 3.1.2 leaves $x,y$ unchanged but decreases $z$ by 1; and then step 3.1.4 decreases $x$ by 1. Thus, the state keeps decreasing during the procedure, meaning that the latter terminates at some point.

(ii) Secondly, we show that, for the fixed $i$, step 3.1 ends up in finite steps. Suppose that, at the beginning of step 3.1 in the $i$th iteration of the outmost loop of Step 3, the number of non-singleton clusters of rank $i$ in the current cirquent $\mathcal{D}$ is $m$. Then, after the four-step procedure terminates for the first time, we get the resulting cirquent $\mathcal{D}_{1}$ where the number of non-singleton clusters of rank $i$ is $m-1$, since no new non-singleton clusters of rank $i$ are created during the four-step procedure and cluster $k$ became a singleton. Pick any cluster $k'$  from the collection of all non-singleton clusters of rank $i$ in  $\mathcal{D}_{1}$ and carry out the same steps as we did with $\mathcal{D}$, then we get the resulting cirquent $\mathcal{D}_{2}$, with $m-2$ non-singleton clusters of rank $i$. Every time the above steps are performed, the number of non-singleton clusters of rank $i$ in the current (topmost) cirquent decreases by 1. Therefore, sooner or later, we get the resulting cirquent $\mathcal{D}_{m}$ having no non-singleton clusters of rank $i$.

(iii) Finally, it's not hard to see that, for the fixed $i$, step 3.2 terminates in finite steps. Hence, the overall Step 3 will end up in finite steps.

Thus, our construction of a proof of $\mathcal{C}$ continues upward from $\mathcal{C}_2$ to the resulting cirquent $\mathcal{C}_3$ after applying Step 3. Since only Rule II, Rule III and Rule IV are applied during Step 3 when we get $\mathcal{C}_3$ from $\mathcal{C}_2$ and all of these rules preserve truth in the bottom-up direction, $\mathcal{C}_{3}$ is valid. Further, we see that all the clusters of each rank of $\mathcal{C}_3$ are singletons (because,  during Step 3, once all the clusters of any rank $i$ of the current cirquent become singletons, the later steps will not change them into non-singletons). Therefore, $\mathcal{C}_{3}$ is classical.
Hence, our construction of a proof of $\mathcal{C}$ ends up with the top most cirquent $\mathcal{C}_{3}$, which is an axiom of $RIF_p$.
\end{proof}

\end{document}